\documentclass[letterpaper, 10 pt, conference]{ieeeconf} 
\IEEEoverridecommandlockouts
\usepackage{graphics}
\usepackage{amsmath}
\usepackage{amssymb}
\usepackage{hyperref}
\usepackage[dvipsnames]{xcolor}
\usepackage{subcaption}
\usepackage{algpseudocode}
\usepackage{algorithm}
\usepackage{centernot}

\newtheorem{definition}{Definition}  
\newtheorem{theorem}{Theorem}        
\newtheorem{proposition}{Proposition}
\newtheorem{lemma}{Lemma}            
\newtheorem{corollary}{Corollary}    

\usepackage{tikz}
\usetikzlibrary{automata, arrows, arrows.meta, positioning, decorations.pathreplacing, calc} 

\tikzset{
    node split radius/.initial=1,
    node split color 1/.initial=red,
    node split color 2/.initial=green,
    node split color 3/.initial=blue,
    node split half/.style={node split={#1,#1+180}},
    node split/.style args={#1,#2}{
        path picture={
            \tikzset{
                x=($(path picture bounding box.east)-(path picture bounding box.center)$),
                y=($(path picture bounding box.north)-(path picture bounding box.center)$),
                radius=\pgfkeysvalueof{/tikz/node split radius}
            }
            \foreach \ang[count=\iAng, remember=\ang as \prevAng (initially #1)] in {#2,360+#1}
                \fill[line join=round, draw, fill=\pgfkeysvalueof{/tikz/node split color \iAng}]
                (path picture bounding box.center)
                --++(\prevAng:\pgfkeysvalueof{/tikz/node split radius})
                arc[start angle=\prevAng, end angle=\ang] --cycle;
        }
    }
}

\tikzset{
    every state/.style={draw, shape=circle, minimum size=18pt, inner sep=1pt}, 
    initial text={},
    c/.style={->},
    u/.style={->, densely dashed},
    state0/.style={state, fill=Gray},
    state1/.style={state, fill=SkyBlue},
    state2/.style={state, fill=Apricot},
    state12/.style={state, node split color 1=SkyBlue, node split color 2=Apricot, node split half=90},
    state3/.style={state, fill=Plum},
}

\newcommand{\buildtrans}[4][]{%
\begin{tikzpicture}[anchor=base,inner sep=0pt,scale=#2,transform shape]
\path (0,0) edge[#4,#1]%
node[above,outer sep=2pt,inner sep=0pt,pos=.45]{\scriptsize$#3$}%
node[below,outer sep=2pt,pos=.45]{\scriptsize~} (18pt,0);%
\end{tikzpicture}}

\def\trans#1{\ensuremath{\mathbin{\kern1pt\buildtrans{0.8}{#1}{->}}}}
\def\transs#1{\ensuremath{\mathbin{\kern1pt\buildtrans{0.8}{#1}{->>}}}}
\def\dtrans#1{\ensuremath{\mathbin{\kern1pt\buildtrans[dashed]{0.8}{#1}{->}}}}
\def\dtranss#1{\ensuremath{\mathbin{\kern1pt\buildtrans[dashed]{0.8}{#1}{->>}}}}

\usepackage{xspace}

\newcommand{\mode}{\text{$M$}\xspace}
\newcommand{\mmode}{\text{$\mathcal{M}$}\xspace}
\newcommand{\brs}{\text{$\mathrm{BRS}$}\xspace}
\newcommand{\funcstyle}[4]{\text{$\mathbf{#1}^{#2}_{#3}(#4)$}\xspace}
\newcommand{\re}  [1]{\funcstyle{R}{}{}{#1}}
\newcommand{\bad} [1]{\funcstyle{B}{}{}{#1}}
\newcommand{\snb} [1]{\funcstyle{N}{}{}{#1}}
\newcommand{\mnb} [2]{\funcstyle{MN}{}{#1}{#2}}

\newcommand{\mmnb}[1]{\funcstyle{MMN}{}{\mathcal{M}}{#1}}
\newcommand{\gnb} [2]{\funcstyle{GN}{}{#1}{#2}}
\newcommand{\pnb} [1]{\funcstyle{PN}{}{\Sigma_{p}}{#1}}
\newcommand{\mtnb}[1]{\funcstyle{MTN}{}{\mathcal{C}}{#1}}
\newcommand{\qnb} [1]{\funcstyle{QN}{}{k}{#1}}

\newcommand{\automaton}{$( Q, \, \allowbreak \Sigma, \, \allowbreak \trans{}, \, \allowbreak Q_{m}, \, \allowbreak q_{0} )$\xspace}
\newcommand{\supervisor}{$( Q_{S}, \, \allowbreak \Sigma_{S}, \, \allowbreak \trans{}_{S}, \, \allowbreak Q_{m,S}, \, \allowbreak q_{0,S} )$\xspace}
\newcommand{\supervisorconstruct}{$( Q_{S}, \, \allowbreak \Sigma, \, \allowbreak \trans{} \cap \allowbreak (Q_{S} \times \allowbreak \Sigma \times \allowbreak Q_{S}), \, \allowbreak Q_{m} \cap Q_{S}, \, \allowbreak q_{0} \cap Q_{S} )$\xspace}

\usepackage{todonotes}
\title{\LARGE \bf
Multimodal Nonblocking Supervisory Control Synthesis
}

\author{M. Minkenberg$^{1}$, M. A. Reniers$^{1}$, M. A. Goorden$^{1}$, J. M. van de Mortel-Fronczak$^{1}$, W. J. Fokkink$^{2}$%
\thanks{*This work was supported as part of STORM\_SAFE, an Interreg North Sea project co-funded by the European Union.}%
\thanks{$^{1}$Department of Mechanical Engineering, Eindhoven University of Technology, Eindhoven, the Netherlands. 
\texttt{\{\href{mailto:m.minkenberg@tue.nl}{m.minkenberg}, \allowbreak
\href{mailto:m.a.reniers@tue.nl}{m.a.reniers}, \allowbreak
\href{mailto:m.a.goorden@tue.nl}{m.a.goorden}, \allowbreak
\href{mailto:j.m.v.d.mortel@tue.nl}{j.m.v.d.mortel}\}@tue.nl} }%
\thanks{$^{2}$Department of Computer Science, Vrije Universiteit Amsterdam, Amsterdam, the Netherlands.
\texttt{\href{mailto:w.j.fokkink@vu.nl}{w.j.fokkink@vu.nl}}%
}
}

\usepackage{eso-pic}
\AddToShipoutPictureBG*{%
\AtPageUpperLeft{%
\setlength\unitlength{1in}%
\hspace*{\dimexpr0.5\paperwidth\relax}
\makebox(0,-1.25)[c]{
\begin{tabular}{c c}
Marijn Minkenberg \emph{et al.},
Multimodal Nonblocking Supervisory Control Synthesis \\
Submitted to
{\em CDC 2026},
uploaded to arXiv on \today. \\
\end{tabular}}}}

\begin{document}

\maketitle
\thispagestyle{empty}
\pagestyle{empty}

\pagestyle{plain}


\begin{abstract}
Supervisory control synthesis leverages the nonblocking property to show liveness of the supervised system. This property is particularly weak when system models include fault behavior, reconfiguration, or multiple control goals.
To capture a more suitable nonblocking property for such system models, this paper introduces modal and multimodal nonblocking. These novel nonblocking variants impose a restriction on the states visited on the path towards a marked state. 
Synthesis algorithms are presented to construct modal and multimodal nonblocking supervisors. The novel nonblocking variants are illustrated with three intuitive examples, inspired by real synthesis problems encountered while applying supervisory control synthesis to safety-critical water infrastructures. A comparison is made between the novel nonblocking variants and established nonblocking variants to show that they are distinct. Additionally, where possible, conditions are formulated under which one variant implies the other.
\end{abstract}


\section{Introduction} \label{sec:introduction}

Supervisory control synthesis is a method for the calculation of a correct-by-construction supervisor based on a plant model~\cite{Ramadge1987SupervisoryProcesses}.
It relies on \emph{nonblocking} to show liveness of the supervised system. A supervised system is called nonblocking if it is always able to reach a marked state, often indicating the `rest state' or `safe state' of a system~\cite{Cassandras2008}.
However, nonblocking relies on the assumption that the specific path that is used to reach a marked state does not matter. 
For plants that deal with fault behavior, are reconfigurable, or address multiple control goals, this assumption does not hold.

Over the years, several nonblocking variants have been introduced to obtain more suitable nonblocking properties for certain applications. For example, generalized nonblocking~\cite{Malik2008GeneralisedNonblocking} defines a subset of states from which a marked state must be reachable. This is relevant for hierarchical supervisory control.
Nonblocking with progressive events~\cite{Ware2014ProgressiveVerification} (henceforth progressive nonblocking) restricts the set of events that may be used to reach a marked state. This is particularly useful when the plant contains rare or undesirable behavior.
Multitasking nonblocking~\cite{DeQueiroz2005MultitaskingSystems} introduces colored marked states, 
capturing that multiple rest states must always remain reachable. This is a suitable nonblocking variant for multitasking systems.
Quantitative nonblocking~\cite{Zhang2024QuantitativelySystem} places a restriction on the number of events that may be taken to reach a marked state. This is suitable when such a metric must be limited, for example for workpiece-processing plants.

In ongoing research~\cite{Minkenberg2026AbstractionInsights}, supervisory control synthesis is being applied to safety-critical water infrastructures.
These systems address various---and often conflicting---control goals, like flood protection, water discharge, fish migration, and ship traffic. Each goal typically has its own set of marked states. The active goals are reconfigured based on environmental parameters such as water levels, the time of year, and the presence of ships.

For these systems, it makes sense to synthesize a supervisor that is always able to reach a marked state without changing the active set of control goals. After all, the system has little to no influence over when or whether the control goals change.
In addition, for critical infrastructures, it is essential to include error states and have the supervisor deal with fault behavior. Naturally, the supervisor should not rely on error states to reach a marked state.

Synthesizing such a supervisor requires a nonblocking variant that restricts the states visited on the path to a marked state. This precise control problem is not addressed by established nonblocking variants. Multitasking nonblocking captures the reachability of multiple marked states at once, but it does not restrict how any of those marked states are reached. Progressive nonblocking does place restrictions on the path to a marked state. However, the restriction is expressed in events and can only be applied for a single configuration or control goal at a time.

In general, it is desirable to be able to express nonblocking without relying on fault behavior, reconfiguration, or different control goals. It is an easy enough modeling task to define a set of states belonging to a certain configuration, or a set of states without a system error. However, there is no method to synthesize a supervisor that guarantees nonblocking without ever leaving such a state set.

This paper introduces \emph{modal nonblocking} and \emph{multimodal nonblocking} as novel nonblocking variants. They enable the modeler to specify one or more state sets and show that a marked state can be reached without leaving the set(s). These variants provide a more expressive and versatile nonblocking property that can be used for the aforementioned control problem with safety-critical water infrastructures.

Additionally, algorithms are presented for synthesizing modal and multimodal nonblocking supervisors. The computational complexity of the algorithms scales linearly with the number of defined state sets. The synthesized supervisors are shown to be controllable and maximally permissive.

Finally, the novel nonblocking variants are compared to previously established ones to show that they are distinct. Where possible, specific conditions are given under which one nonblocking variant implies the other.


\section{Preliminaries} \label{sec:preliminaries}

\subsection{Finite automata}
This paper uses \emph{finite automata}, as introduced in \cite{Arnold1994FiniteSystems}.

\begin{definition}[Finite Automaton]
A \emph{finite automaton} (FA) is a 5-tuple  $P = $ \automaton, where:
\begin{itemize}
    \item $Q$ is a finite set of states;
    \item $\Sigma$ is a finite set of events, partitioned into controllable events $\Sigma_c$ and uncontrollable events $\Sigma_{u}$;
    \item $\trans{} \subseteq Q \times \Sigma \times Q$ is the set of transitions;
    \item $Q_{m} \subseteq Q$ is the set of marked states;
    \item $q_{0} \in Q$ is the initial state.
\end{itemize}
\end{definition}

An FA $P_{1} = ( Q_{1}, \, \allowbreak \Sigma, \, \allowbreak \trans{}_{1}, \, \allowbreak Q_{m,1}, \, \allowbreak q_{0} )$ is a \emph{subautomaton} of FA $P_{2} = ( Q_{2}, \, \allowbreak \Sigma, \, \allowbreak \trans{}_{2}, \, \allowbreak Q_{m,2}, \, \allowbreak q_{0} )$, denoted $P_{1} \sqsubseteq P_{2}$, iff $Q_{1} \subseteq Q_{2}$, $\trans{}_{1} \subseteq \trans{}_{2}$, and $Q_{m,1} = Q_{m,2} \cap Q_{1}$.
In other words, $P_{1} \sqsubseteq P_{2}$ if $P_{1}$ can be obtained from $P_{2}$ by removing transitions and/or states, where all transitions that originate from or go to removed states are also removed \cite{Ouedraogo2010SymbolicAutomata}.

A transition $(q, \, \sigma, \, q') \in \trans{}$, also written as $q \trans{\sigma} q'$, defines that $P$ is able to go from state $q$ to state $q'$ by executing event $\sigma$. The automata in this paper are \emph{deterministic}, meaning that $q \trans{\sigma} q'$ and $q \trans{\sigma} q''$ always implies $q' =  q''$. 
The following notations are used:
\begin{itemize}
    \item $q \trans{} q'$ denotes that $q \trans{\sigma} q'$ for some $\sigma \in \Sigma$;
    \item $q \dtrans{} q'$ denotes that $q \trans{\sigma_{u}} q'$ for some $\sigma_{u} \in \Sigma_{u}$;
    \item $q \trans{}^{-1} q'$ denotes that $q' \trans{} q$;
    \item $q \transs{} q'$ denotes that $q \trans{} q''$ and $q'' \transs{} q' $ for some $q''$, or $q = q'$, and this applies to $\dtranss{}$ and $\transs{}_{x}$ in the same way.
\end{itemize}

\begin{definition}[Reachable]
A state $q \in Q$ is considered \emph{reachable}, denoted \re{q}, iff $q_{0} \transs{} q$. An FA $P = $ \automaton is \emph{reachable} iff $\forall q \in Q (\re{q})$.
\end{definition}


\subsection{Supervisory control synthesis}

Supervisory control synthesis~\cite{Ramadge1987SupervisoryProcesses} constructs a supervisor FA $S$ based on a plant FA $P$.
In this paper, it is assumed that $S \sqsubseteq P$ and $S$ is reachable.
The supervisor interacts with the plant by disabling some controllable events.
The synthesis process must construct a supervisor that is \emph{nonblocking}, \emph{controllable}, and \emph{maximally permissive}.

The nonblocking property of a supervisor depends on the nonblocking properties of its states.

\begin{definition}[State properties]
\label{def:state-properties}
Let $P = $ \automaton be an FA. A state $q \in Q$ is:
\begin{itemize}
    \item \emph{nonblocking}, denoted \snb{q}, iff $\exists q_{m} \in Q_{m} \, (q \transs{} q_{m})$;
    \item \emph{blocking} iff $\neg \snb{q}$;
    \item \emph{bad}, denoted \bad{q}, iff $\exists q' \in Q \, (\neg \snb{q'} \land q \dtranss{} q')$.
\end{itemize}
\end{definition}
 
A supervisor is \emph{nonblocking} if all reachable states in the supervisor are nonblocking states.

\begin{definition}[Standard nonblocking~\cite{Cassandras2008}]
\label{def:nonblocking-FA}
Let $P = $ \automaton be an FA. $P$ is \emph{standard nonblocking} or simply \emph{nonblocking}, denoted \snb{P}, iff $\forall q \in Q \, (\re{q} \Rightarrow \snb{q})$.
\end{definition}

A supervisor is \emph{controllable} for a plant if it never disables uncontrollable events in $\Sigma_{u}$ that are possible in the plant.

\begin{definition}[Controllable~\cite{Thuijsman2022}]
\label{def:controllable}
Let $P = $ \automaton be a plant FA, let $S = $ \supervisor be a supervisor FA, and let $\Sigma_{u} \subseteq \Sigma$ be the set of uncontrollable events. $S$ is \emph{controllable} for $P$ iff $\trans{} \cap (Q_{S} \times \Sigma_{u} \times Q) \subseteq \trans{}_{S}$.
\end{definition}

A supervisor that is \emph{nonblocking} and \emph{controllable} is called \emph{proper}. Finally, a supervisor is \emph{maximally permissive} if it does not disable more behavior than necessary to be proper.

\begin{definition}[Maximally permissive]
\label{def:maximally-permissive}
Let $P$ be a plant FA and $S$ be a proper supervisor FA. $S$ is \emph{maximally permissive} iff for each proper supervisor FA $S'$ it holds that $S' \sqsubseteq S$.
\end{definition}


\subsection{Nonblocking property}
To obtain the set of nonblocking states, a \textsc{Backward Reachability Search} is defined, see Algorithm~\ref{alg:reach}. Its format and use in synthesis are inspired by~\cite{Thuijsman2022}.

\begin{algorithm}
\caption{\textsc{Backward Reachability Search} (\brs).} \label{alg:reach}
\begin{algorithmic}[1]

\item[] \textbf{Input:} Considered state set $Q_{c}$,
starting state set $Q_{s}$, transition relation $\trans{}$
\item[] \textbf{Output:} Maximal state set $Q_{r}\subseteq Q_{c}$ from which a state $q \in Q_{c} \cap Q_{s}$ is reachable with transitions in $\trans{}$
\State $i \gets 0, \, Q^{i} \gets Q_{c} \cap Q_{s}$
\Repeat
\State $Q^{i+1} \gets Q^{i} \cup \{ q \in Q_{c} | q \trans{} q', q' \in Q^{i}\}$
\State $i \gets i + 1$
\Until $Q^{i} = Q^{i-1}$
\State $Q_{r} \gets Q^{i}$

\end{algorithmic}
\end{algorithm}

This algorithm takes all states in the set $Q_{c} \cap Q_{s}$, and iteratively finds more states in $Q_{c}$ that have a sequence of transitions in $\trans{}$ to a state in this set.
The loop terminates when no new states are found. At that point, $Q^{i}$ contains all states that can reach $Q_{c} \cap Q_{s}$ with transitions in $\trans{}$.

\brs can be used to find the set of nonblocking states, the set of bad states, and the set of reachable states. 
This is used for synthesis, as presented in Algorithms~\ref{alg:synthesis-modal-nonblocking} and~\ref{alg:synthesis-multimodal-nonblocking}.
\begin{lemma}\cite{Thuijsman2022}.
\label{lem:use-breach}
Let $P = $ \automaton be an FA, and let $Q_{\mathrm{blocking}} = \{q \in Q | \neg \snb{q}\}$ be the set of blocking states of $P$. Then, the following holds:
\begin{itemize}
    \item $\forall q \in Q \, (\snb{q} \Leftrightarrow q \in \brs(Q, \, Q_{m}, \, \trans{}))$;
    \item $\forall q \in Q \, (\bad{q} \Leftrightarrow q \in \brs(Q, \, Q_{\mathrm{blocking}}, \, \dtrans{}))$;
    \item $\forall q \in Q \, (\re{q} \Leftrightarrow q \in \brs(Q, \, \{q_{0}\}, \, \trans{}^{-1}))$. 
\end{itemize}
\end{lemma}


\section{Motivating examples} \label{sec:examples}

This section features three examples to illustrate different weaknesses of nonblocking. Each example
is also encountered when applying supervisory control synthesis to safety-critical water infrastructures. For each example, it is shown how a more suitable nonblocking property can be obtained.

\subsection{Example 1: System with error states}
Consider an FA that includes error states, such as the one shown in Fig.~\ref{fig:example-faults}. Here, states $C$ and $D$ represent error states. The FA is nonblocking, but this relies on states $C$ and $D$ to be visited. It is argued that a stronger nonblocking property is necessary that determines nonblocking without relying on these error states.

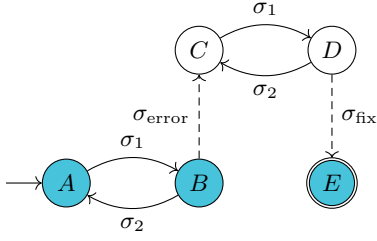
\begin{figure}[!ht]
    \centering
    \begin{tikzpicture}[node distance=5.0em, font=\small]
    
        \node [state1, initial left] (A) {$A$};
        \node [state1, right of=A] (B) {$B$};
        \node [state, above of=B] (C) {$C$};
        \node [state, right of=C] (D) {$D$};
        \node [state1, below of=D, accepting] (E) {$E$};
        
        \draw (A) edge [c, bend left] node [above] {$\sigma_{1}$} (B);
        \draw (B) edge [c, bend left] node [below] {$\sigma_{2}$} (A);
        \draw (B) edge [u] node [left] {$\sigma_{\mathrm{error}}$} (C);
        \draw (C) edge [c, bend left] node [above] {$\sigma_{1}$} (D);
        \draw (D) edge [c, bend left] node [below] {$\sigma_{2}$} (C);
        \draw (D) edge [u] node [right] {$\sigma_{\mathrm{fix}}$} (E);
        
    \end{tikzpicture}
    \caption{Example FA for a system with error states.}
    \label{fig:example-faults}
\end{figure}

Suppose that a subset of states is chosen to describe an error-free mode. Here, this is the set $\{ A,\, B,\, E \}$, as highlighted in blue. It is then clear that, from states $A$ and $B$, there exists no path towards a marked state without leaving this set. In other words, the system becomes blocking if it is restricted to the error-free mode.

By defining the error-free mode, it becomes possible to express that this FA is only nonblocking when it is allowed to leave the error-free mode. The modeler might want to reconsider whether this behavior is desirable.


\subsection{Example 2: System with reconfiguration}

The study in \cite{Thuijsman2024SupervisoryLines} introduces a framework for modeling product lines with reconfigurable features. Their running example is a coffee machine where various features may \emph{come} or \emph{go} at any moment, for example as a result of maintenance. Although the work shows that a nonblocking supervisor can be synthesized for the overall system, it has not been considered whether nonblocking holds for each individual configuration. Because the system may operate for a long time in one configuration, it is desirable that each configuration is nonblocking when considered in isolation.

A small FA with reconfiguration is shown in Fig.~\ref{fig:example-dynamic-feature-configuration}. Here, reconfiguration events $\sigma_{\mathrm{X\_come}}$ and $\sigma_{\mathrm{X\_go}}$ represent some feature $\mathrm{X}$ coming or going from the system.
In this example, states $B_{1}$ and $B_{2}$ represent a mode where feature $\mathrm{X}$ is present in the system. Because neither of these two states is marked, the system is never in a marked state while $\mathrm{X}$ is present, which may be indefinitely.
However, the FA is nonblocking because marked state $P_{2}$ is always reachable. The nonblocking property relies on ever returning to states where feature $\mathrm{X}$ is not present.

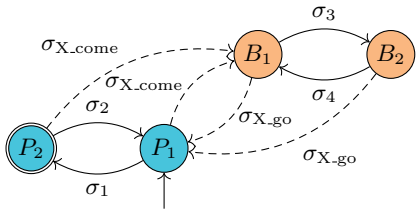
\begin{figure}[!b]
    \centering
    \begin{tikzpicture}[node distance=5.0em, font=\small]

        \node [state1, initial below] (A1) {$P_{1}$};
        \node [state1, accepting, left of=A1] (A2) {$P_{2}$};
        \node [state2, above right of=A1] (B1) {$B_{1}$};
        \node [state2, right of=B1] (B2) {$B_{2}$};
        
        \draw (A1) edge [c, bend left] node [below] {$\sigma_{1}$} (A2);
        \draw (A2) edge [c, bend left] node [above] {$\sigma_{2}$} (A1);
        \draw (A1) edge [u, bend left] node [left] {$\sigma_{\mathrm{X\_come}}$} (B1);
        \draw (B1) edge [u, bend left] node [right] {$\sigma_{\mathrm{X\_go}}$} (A1);
        \draw (B1) edge [c, bend left] node [above] {$\sigma_{3}$} (B2);
        \draw (B2) edge [c, bend left] node [below] {$\sigma_{4}$} (B1);
        \draw (A2) edge [u, bend left] node [above left] {$\sigma_{\mathrm{X\_come}}$} (B1);
        \draw (B2) edge [u, bend left] node [below right] {$\sigma_{\mathrm{X\_go}}$} (A1);
        
    \end{tikzpicture}
    \caption{Example FA for a system with reconfiguration.}
    \label{fig:example-dynamic-feature-configuration}
\end{figure}

One way to reveal this issue is to consider nonblocking separately for each configuration. In the example, the configurations are given by state sets $\{P_{1}, \, P_{2}\}$ and $\{B_{1}, \, B_{2}\}$, highlighted in blue and orange, respectively. Both states in set $\{P_{1}, \, P_{2}\}$ can reach a marked state while staying in the set. However, neither state in set $\{B_{1}, \, B_{2}\}$ can reach a marked state while staying in the set. As a result, $B_{1}$ and $B_{2}$ are blocking states within this set---a clear indication that there is an issue in this FA, whenever $\mathrm{X}$ is present.


\subsection{Example 3: System with multiple control goals}

Consider an FA that describes two controlled traffic lights at a crossing, shown in Fig.~\ref{fig:example-trafficlights}. The state names indicate for each traffic light whether they are red ($R$) or green ($G$). The FA leaves out the $GG$ state, where both traffic lights are green simultaneously.
Clearly, this FA is nonblocking, which guarantees that at least one traffic light can turn green. However, nonblocking alone does not guarantee that each traffic light can turn green when their behavior is isolated.

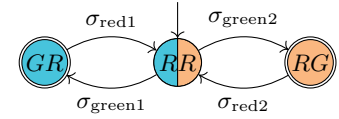
\begin{figure}[!ht]
    \centering
    \begin{tikzpicture}[node distance=5.0em, font=\small]
        
        \node [state12, initial above] (A) {$RR$};
        \node [state1, accepting, left of=A] (C) {$GR$};
        \node [state2, accepting, right of=A] (B) {$RG$};
        
        \draw (A) edge [c, bend left] node [above] {$\sigma_{\mathrm{green2}}$} (B);
        \draw (B) edge [c, bend left] node [below] {$\sigma_{\mathrm{red2}}$}  (A);
        \draw (A) edge [c, bend left] node [below] {$\sigma_{\mathrm{green1}}$}  (C);
        \draw (C) edge [c, bend left] node [above] {$\sigma_{\mathrm{red1}}$}  (A);
        
    \end{tikzpicture}
    \caption{Example FA for a system with multiple control goals.}
    \label{fig:example-trafficlights}
\end{figure}

Suppose a subset of states is chosen that isolates the behavior of the first traffic light (the set $\{ RR, \, GR \}$, as highlighted in blue). If the FA is nonblocking \emph{within} this set, it is known more specifically that the first traffic light can turn green without relying on the other traffic light. The same is possible separately for the second traffic light by choosing the set $\{ RR, \, RG \}$, as highlighted in orange.


\section{Modal nonblocking} \label{sec:modal-nonblocking}

In each of the previous examples, a more suitable nonblocking property is obtained by restricting the set of states that may be visited on the path to a marked state. This section formalizes that method by introducing the novel concept of modal nonblocking.

\subsection{Definition of modal nonblocking}
Modal nonblocking is based on a parameter $\mode \subseteq Q$, which is a subset of plant states, referred to as a mode. A mode may represent error-free states, states of a system configuration, or states related to a specific control goal.
For a state $q \in \mode$, modal nonblocking requires that it is possible to reach a marked state without leaving \mode. In other words, any state $q \centernot\in \mode$ should not be relied upon to reach a marked state. A state $q \centernot\in \mode$ itself is defined as modal nonblocking by default.

\begin{definition}[Modal state properties]
\label{def:modal-nonblocking-state}
Let $P = $ \automaton be an FA, let $\mode \subseteq Q$ be a state set of interest, and let $\trans{}_{\mode} = \trans{} \cap \ \mode \times \Sigma \times \mode$ be the set of transitions \emph{within} \mode. A state $q \in Q$ is:

\begin{itemize}
    \item \emph{modal nonblocking} for \mode, denoted \mnb{\mode}{q}, iff $q \in \mode \Rightarrow \exists q_{m} \in Q_{m} \, (q \transs{}_{\mode} q_{m})$;
    \item \emph{modal blocking} for \mode iff $\neg \mnb{\mode}{q}$;
    \item \emph{bad}, denoted \bad{q}, iff $\exists q' \in Q \, (\neg \mnb{\mode}{q'} \land q \dtranss{} q')$.
\end{itemize}
\end{definition}

Note that $\trans{}_{\mode}$ is introduced as the set of transitions within \mode. If there exists a sequence of transitions within \mode, then that sequence also exists in general.

\begin{lemma}\label{lem:transs-transsM}
Let $P = $ \automaton, let $q$ and $q'$ be states in $Q$, let $\mode \subseteq Q$ be a state set of interest, and let $\trans{}_{\mode} = \trans{} \cap \mode \times \Sigma \times \mode$.
Then, $q \transs{}_{\mode} q' \Rightarrow q \transs{} q'$.
\end{lemma}

\begin{proof}
Suppose there exists a sequence of transitions from $q$ to $q'$ in $\transs{}_{\mode}$. Every transition in this sequence must also be in $\trans{}$, because by the construction of $\trans{}_{\mode}$ it holds that $\trans{}_{\mode} \subseteq \trans{}$. Therefore, the whole sequence of transitions is also in $\transs{}$. Thus, $q \transs{}_{\mode} q' \Rightarrow q \transs{} q'$.
\end{proof}

Similar to finding the set of nonblocking, bad, and reachable states (see Lemma~\ref{lem:use-breach}), the set of modal nonblocking states in \mode can be found using \brs.
\begin{lemma}
\label{lem:mnb-breach}
Let $P = $ \automaton be an FA, and let $\mode \subseteq Q$ be a state set of interest. Then $\forall q \in \mode \, (\mnb{\mode}{q} \Leftrightarrow q \in \brs(\mode, \, Q_{m}, \, \trans{}))$. 
\end{lemma}

\begin{proof}
This proof is structured over the lines in \brs where $Q^{i}$ is constructed, which becomes the output $Q_{r}$ on line~6.
\begin{itemize}
\item[1:] $Q^{0}$ is constructed as $\mode \cap Q_{m}$. Thus, $Q^{0}$ is the set of states in \mode that can reach a state in $\mode \cap Q_{m}$ using 0 transitions in $\trans{}$, while staying in \mode.

\item[3:] $Q^{i+1}$ is constructed as $Q^{i}$ extended with the set of states in \mode that can reach a state in $Q^{i}$ after 1 transition in $\trans{}$. Since $Q^{i} \subseteq \mode$, both $q$ and $q'$ are in \mode. Thus, $Q^{i+1}$ is the set of states in \mode that can reach a state in $\mode \cap Q_{m}$ using $i+1$ or less transitions in $\trans{}$, while staying in \mode.

\item[5:] The loop finishes when $Q^{i} = Q^{i-1}$. Because $Q^{i-1} \subseteq Q^{i}$, the set of states $Q^{i}$ only grows. Since $Q^{i} \subseteq \mode$ and \mode is finite, the expansion of $Q^{i}$ must be finite.
As a result, $Q_{r}$ is the set of all states in \mode that can reach a state in $\mode \cap Q_{m}$ using any number of transitions in $\trans{}$, while staying in \mode.
\end{itemize}
At every step of its construction, it holds that \mnb{\mode}{q} for all states in $Q^{i}$. The output of \brs, $Q_{r}$, is the set of all states in \mode for which \mnb{\mode}{q} holds. Thus, for $q \in \mode$, it is proven that $q \in \brs( \mode, \, Q_{m}, \, \trans{} ) \Leftrightarrow \mnb{\mode}{q}$.
\end{proof}

The modal nonblocking property for an FA requires that all its reachable states are modal nonblocking. It does not pose additional restrictions on the reachability of a state.

\begin{definition}[Modal nonblocking FA]
\label{def:modal-nonblocking-FA}
Let $P = $ \automaton be an FA, and let $\mode \subseteq Q$ be a state set of interest. $P$ is \emph{modal nonblocking} for \mode, denoted \mnb{\mode}{P}, iff $\forall q \in Q \, (\re{q} \Rightarrow \mnb{\mode}{q})$.
\end{definition}

The definition is illustrated using Example~1, where $\mode = \{ A, \, B, \, E\}$. All states in the FA are reachable. States $C$ and $D$ are not in \mode, and therefore they are modal nonblocking by default. State $E$ is modal nonblocking, since it is marked. States $A$ and $B$ are not modal nonblocking: both states are in \mode and cannot reach a marked state without visiting a state outside of \mode. Because not all reachable states are modal nonblocking, the FA is also not modal nonblocking.


\subsection{Modal nonblocking synthesis}
The goal of synthesis with modal nonblocking is to produce a supervisor that is modal nonblocking, controllable, and maximally permissive.
Algorithm~\ref{alg:synthesis-modal-nonblocking} presents a synthesis algorithm for this purpose, relying on the \brs function from Algorithm~\ref{alg:reach} to keep the presentation compact.

Algorithm~\ref{alg:synthesis-modal-nonblocking} is a small adjustment of the standard synthesis algorithm used in~\cite{Thuijsman2022}, where the changes are highlighted in blue.
Note that this algorithm is not optimized for efficiency; however, the computational complexity of modal nonblocking is equal to that of synthesis with the standard nonblocking property.

\begin{algorithm}
\caption{Synthesis with modal nonblocking.} \label{alg:synthesis-modal-nonblocking}
\begin{algorithmic}[1]
\item[] \textbf{Input: } plant \automaton \textcolor{blue}{and $\mode \subseteq Q$}
\item[] \textbf{Output: } supervisor \supervisorconstruct
\State $k \gets 0, Q^{k} \gets Q$
\Repeat
    \State \textcolor{blue}{$Q^{k}_{\mode} \gets Q^{k} \cap \mode$}
    \State $N^{k} \gets \brs(Q_{\textcolor{blue}{\mode}}^{k}, \, Q_{m}, \, \trans{})$
    \State $B^{k} \gets \brs(Q^{k}, \, Q_{\textcolor{blue}{\mode}}^{k} \setminus N^{k}, \, \dtrans{})$
	\State $Q^{k+1} \gets Q^{k} \setminus B^{k}$
	\State $k \gets k+1$
\Until $Q^{k} = Q^{k-1}$
\State $Q_{S} \gets$ \brs($Q^{k}, \, \{q_{0}\}, \, \trans{}^{-1}$)
\end{algorithmic}
\end{algorithm}

The newly introduced input set \mode is used on line~3 to specify the set of states $Q^{k}_{\mode}$ for which modal nonblocking should be determined. States outside of \mode are always modal nonblocking and are therefore excluded from this set.
The set of modal nonblocking states is obtained on line~4 using Lemma~\ref{lem:mnb-breach}.
The set of modal blocking states is obtained by subtracting the set of modal nonblocking states $N^{k}$ from $Q^{k}_{\mode}$, because only states in $Q^{k}_{\mode}$ can become modal blocking.
The set of bad states $B^{k}$ is obtained on line~5, using the second item of Lemma~\ref{lem:use-breach}. Here, the considered state set is $Q^{k}$, because any state in $Q^{k}$ can become a bad state: this must not be limited to only states in $Q^{k}_{\mode}$.

After the computation of bad states $B^{k}$, they are removed from $Q^{k}$ on line~6. Thus, the algorithm iteratively refines the set $Q^{k}$, each time taking out any blocking states, until $Q^{k}$ is the same in two consecutive iterations.
Finally, \brs is used one more time on line~9 to obtain the set of reachable states, using the third item of Lemma~\ref{lem:use-breach}.
Because at the end all states in $Q^{k}$ are modal nonblocking, it also holds that all reachable states are modal nonblocking. Thus, 
a modal nonblocking FA is constructed.

In fact, Algorithm~\ref{alg:synthesis-modal-nonblocking} constructs a supervisor that is modal nonblocking for \mode, controllable, and maximally permissive.
However, instead of proving the correctness of Algorithm~\ref{alg:synthesis-modal-nonblocking} directly, the full proof is given in Theorem~\ref{thm:algo-multimodal} for synthesis with multimodal nonblocking, of which synthesis with modal nonblocking is a special case. 


\section{Multimodal nonblocking} \label{sec:multimodal-nonblocking}

Examples~2 and~3 have shown the use of modal nonblocking with respect to multiple state sets of interest. This motivates the novel notion of multimodal nonblocking, which is introduced and defined in this section.

\subsection{Definition of multimodal nonblocking}
Multimodal nonblocking is defined for multiple state sets of interest, collected in the set \mmode, where $\mmode \subseteq 2^{Q}$. 

\begin{definition}[Multimodal state properties]
\label{def:multimodal-nonblocking-state}
Let $P = $ \automaton be an FA, and let $\mmode \subseteq 2^{Q}$.
A state $q \in Q$ is:
\begin{itemize}
    \item \emph{multimodal nonblocking} for \mmode, denoted \mmnb{q}, iff $\forall \mode \in \mmode \, (\mnb{\mode}{q})$;
    \item \emph{multimodal blocking} for \mmode iff $\neg\mmnb{q}$;
    \item \emph{bad}, denoted \bad{q}, iff $\exists q'\in Q \, (\neg\mmnb{q'} \land q \dtranss{} q')$.
\end{itemize}
\end{definition}

This definition leverages the modal nonblocking definition from Definition~\ref{def:modal-nonblocking-state}, adding that it must hold for every state set of interest in \mmode.
Note that a state that is not contained in any state sets of interest is always multimodal nonblocking. In addition, a marked state is always multimodal nonblocking.

\begin{lemma}
\label{lem:mmnb-marked}
Let $P =$ \automaton and let $\mmode \subseteq 2^{Q}$. Then,
$\forall q \in Q_{m} \, (\mmnb{q})$.
\end{lemma}

\begin{proof}
For all $q \in Q_{m}$ and $\mode \in \mmode$, it is certain that \mnb{\mode}{q}, since $q \in \mode \Rightarrow \exists q_{m} \in Q_{m} \, (q \transs{}_{\mode} q_{m})$ is already true by not taking any transitions. It thus holds that $\mmnb{q}$ for all $q \in Q_{m}$.
\end{proof}

Multimodal nonblocking is suitable for Example~3, where $\mmode = \{\mode_{b}, \, \mode_{o}\}$ (as indicated by the blue and orange colored states in Figure~\ref{fig:example-trafficlights}, respectively). State $RR$ belongs to both $\mode_{b}$ and $\mode_{o}$, so \mmnb{RR} iff $\mnb{M_{b}}{RR} \land \mnb{M_{o}}{RR}$. Observe that indeed $RR$ is able to reach marked state $GR$ while staying within $\mode_{b}$, and $RR$ is also able to reach marked state $RG$ while staying within $\mode_{o}$.
Thus, \mmnb{RR} holds. 
By Lemma~\ref{lem:mmnb-marked}, both marked states $GR$ and $RG$ are multimodal nonblocking. When \mmnb{q} holds for all reachable states $q$ of an FA $P$, the FA itself is multimodal nonblocking.

\begin{definition}[Multimodal nonblocking FA]
\label{def:multimodal-nonblocking-FA}
Let $P = $ \automaton be an FA, and let $\mmode \subseteq 2^{Q}$.
$P$ is \emph{multimodal nonblocking} for \mmode, denoted \mmnb{P}, iff 
$\forall q \in Q \, (\re{q} \Rightarrow \mmnb{q})$.
\end{definition}

This can now be used to show that Example~2 is not multimodal nonblocking for $\mmode = \{\mode_{b}, \, \mode_{o}\}$ (as indicated by the blue and orange colored states in Figure~\ref{fig:example-dynamic-feature-configuration}, respectively). 
All states are reachable, and all states are modal nonblocking for $M_{b}$. However, states $B1$ and $B2$ are not modal nonblocking for $M_{o}$, since there is no marked state that can be reached while staying within $M_{o}$. Therefore, the FA is clearly not multimodal nonblocking for \mmode.


\subsection{Multimodal nonblocking synthesis}

The algorithm for synthesis with multimodal nonblocking is presented in Algorithm~\ref{alg:synthesis-multimodal-nonblocking}, where the changes with respect to Algorithm~\ref{alg:synthesis-modal-nonblocking} are highlighted in blue.  Note that the computational complexity of Algorithm~\ref{alg:synthesis-multimodal-nonblocking} scales linearly with $| \mmode | $.

\begin{algorithm}
\caption{Synthesis with multimodal nonblocking.} \label{alg:synthesis-multimodal-nonblocking}
\begin{algorithmic}[1]
\item[] \textbf{Input: } plant \automaton \textcolor{blue}{and $\mmode \subseteq 2^{Q}$}
\item[] \textbf{Output: } supervisor \supervisorconstruct
\State $k \gets 0, Q^{k} \gets Q$
\Repeat
    \ForAll{\textcolor{blue}{$\mode \in \mmode$}}
        \State $Q^{k}_{\mode} \gets Q^{k} \cap \mode$
        \State $N^{k}_{\textcolor{blue}{\mode}} \gets \brs(Q^{k}_{\mode}, \, Q_{m}, \, \trans{})$
    \EndFor
    \State $B^{k} \gets \brs(Q^{k}, \, \textcolor{blue}{\bigcup_{\mode \in \mmode}[Q^{k}_{\mode} \setminus N^{k}_{\mode}]}, \, \dtrans{})$
	\State $Q^{k+1} \gets Q^{k} \setminus B^{k}$
	\State $k \gets k+1$
\Until $Q^{k} = Q^{k-1}$
\State $Q_{S} \gets \brs(Q^{k}, \, \{q_{0}\}, \, \trans{}^{-1})$
\end{algorithmic}
\end{algorithm}

This algorithm iterates over the state sets \mode defined in \mmode. Each iteration constructs $N^{k}_{\mode}$, describing the set of nonblocking states for each \mode. On line~7, the term $\bigcup_{\mode \in \mmode}[Q^{k}_{\mode} \setminus N^{k}_{\mode}]$ combines the blocking states for each \mode, forming $Q_{\mathrm{blocking}}$.
After that, the synthesis algorithm continues in the same way as Algorithm~\ref{alg:synthesis-modal-nonblocking}.

Algorithm~\ref{alg:synthesis-multimodal-nonblocking} also performs synthesis with modal nonblocking or standard synthesis for specific inputs:
\begin{itemize}
    \item For $\mmode = \{\mode\}$, synthesis with modal nonblocking for \mode is performed, because lines~3 to~6 are executed once for \mode.
    \item For $\mmode = \{Q\}$, standard synthesis is performed, because lines~3 to~6 are executed once for $\mode = Q$.
\end{itemize}

For any plant model and \mmode, Algorithm~\ref{alg:synthesis-multimodal-nonblocking} constructs a supervisor that is multimodal nonblocking for \mmode. Additionally, the controllability and maximal permissiveness is ensured.

\begin{theorem}
\label{thm:algo-multimodal}
Let $P = $ \automaton be an FA with a set of uncontrollable events $\Sigma_{u} \subseteq \Sigma$, and let $\mmode \subseteq 2^{Q}$.
Algorithm~\ref{alg:synthesis-multimodal-nonblocking} constructs a supervisor $S = $ \supervisor that is multimodal nonblocking for \mmode, controllable for $P$, and maximally permissive.
\end{theorem}

\begin{proof}
\textbf{Multimodal nonblocking for \mmode:}
Starting from Definition~\ref{def:multimodal-nonblocking-FA}, it must be proven that $\forall q \in Q_{S} \, ( \re{q} \Rightarrow \mmnb{q} )$. Here, \mmnb{q} can be replaced with Definitions~\ref{def:modal-nonblocking-state} and~\ref{def:multimodal-nonblocking-state} to get: $\forall q \in Q_{S} \, ( \re{q} \Rightarrow \forall \mode \in \mmode \, ( q \in \mode \Rightarrow \exists q_{m} \in Q_{m} \, ( q \transs{}_{\mode} q_{m} ) ) )$.
Let $q \in Q_{S}$ and assume \re{q}. Then, let $\mode \in \mmode$ and assume $q \in \mode$.
It remains to be proven that $\exists q_{m} \in Q_{m} \, ( q \transs{}_{\mode} q_{m} )$.

On line~11, $Q_{S}$ is constructed based on \brs. A property of \brs is that $Q_{r} \subseteq Q_{c}$. Therefore, $Q_{S} \subseteq Q^{k}$, and since $q \in Q_{S}$, also $q \in Q^{k}$.
Suppose the loop terminates when $K = k - 1$, so when $Q^{K+1} = Q^{K}$, following line~10. From line~8, it can be concluded that $B^{K} = \emptyset$, applying $Q^{K+1} = Q^{K}$ and knowing that neither set is empty (since $q \in Q^{K+1}$). $B^{K}$ is the output of \brs on line~7. This output is at least $Q_{c} \cap Q_{s}$, so it is the empty set iff $Q_{c} \cap Q_{s}$ is also the empty set. So, taking the \brs inputs on line~7, it must hold that $Q^{K} \cap ( \bigcup_{\mode \in \mmode} [Q_{\mode}^{K} \setminus N_{\mode}^{K}] ) = \emptyset$.

From line~4, it is known that $Q^{K}_{\mode} \subseteq Q^{K}$ for every \mode. Then, also $Q^{K}_{\mode} \setminus N^{K}_{\mode} \subseteq Q^{K}$ for every \mode. Apply that $A \cap B = A$ when $A \subseteq B$ to obtain $\bigcup_{\mode \in \mmode} [Q_{\mode}^{K} \setminus N_{\mode}^{K}] = \emptyset$. Thus, $Q^{K}_{\mode} \subseteq N^{K}_{\mode}$ for every \mode.

Combining that $q \in \mode$ as well as $q \in Q^{K}$, by line~4 also $q \in Q^{K}_{\mode}$.
Since $Q^{K}_{\mode} \subseteq N^{K}_{\mode}$ for every \mode, it is known that $q \in N^{K}_{\mode}$ for every \mode. $N^{K}_{\mode}$ is the output state set of $\brs(Q^{K}_{\mode}, \, Q_{m}, \, \trans{})$ on line~5. Because $Q^{K}_{\mode} \subseteq \mode$ by line~4, it also holds that $q \in \brs(\mode, \, Q_{m}, \, \trans{})$.
By Lemma~\ref{lem:mnb-breach}, this means \mnb{\mode}{q}. Thus, by Definition~\ref{def:modal-nonblocking-state}, $\exists q_{m} \in Q_{m} \, ( q \transs{}_{\mode} q_{m} )$, which was required for the proof.

\textbf{Controllable:}
Start from Definition~\ref{def:controllable}: $\trans{} \cap (Q_{S} \times \Sigma_{u} \times Q) \subseteq \trans{}_{S}$. This holds if $(q_{S}, \sigma_{u}, q) \in \trans{} \Rightarrow (q_{S}, \sigma_{u}, q) \in \trans{}_{S}$, where $q_{S} \in Q_{S}$, $\sigma_{u} \in \Sigma_{u}$, and $q \in Q$.
Assume $(q_{S}, \sigma_{u}, q) \in \trans{}$, and prove $(q_{S}, \sigma_{u}, q) \in \trans{}_{S}$.

Since $q_{S} \in Q_{S}$, also $q_{S} \in Q^{k}$ by line~11. Suppose $q_{S} \in B^{k}$ for some $k$. Then $q_{S} \centernot\in Q^{k}$ by line~8, which is a contradiction. Therefore, $q_{S} \centernot\in B^{k}$ for any $k$. This means that $\centernot\exists q' \in Q \, (\neg \mmnb{q'} \land q_{S} \dtranss{} q')$: for any multimodal blocking state $q' \in Q$, there is no sequence of uncontrollable events from $q_{S}$ to $q'$.
In other words, for all $q_{S} \in Q_{S}$ any sequence of uncontrollable events leads to a state that is multimodal nonblocking. That state, consequently, is also in $Q_{S}$.
Since $\trans{}_{S} = \trans{} \cap (Q_{S} \times \Sigma \times Q_{S})$, it is certain that $(q_{S}, \sigma_{u}, q) \in \trans{}_{S}$.

\textbf{Maximally permissive:}
Proof is done by contradiction.
Suppose there is a proper supervisor FA $S' = ( Q_{S}', \, \allowbreak \Sigma_{S}, \, \allowbreak \trans{}_{S}', \, \allowbreak Q_{m,S}', \, \allowbreak q_{0,S} ) \sqsubseteq P$ such that $S' \centernot\sqsubseteq S$.
By Definition~\ref{def:maximally-permissive}, this means that either $Q_{S}' \centernot\subseteq Q_{S}$, or $\trans{}_{S}' \centernot\subseteq \trans{}_{S}$, or $Q_{m,S}' \centernot= Q_{m,S} \cap Q_{S}'$.
\begin{itemize}
    \item $Q_{S}' \centernot\subseteq Q_{S}$: Assume $\exists q' \, (q' \in Q_{S}' \land q' \centernot\in Q_{S})$. There are two lines in Algorithm~\ref{alg:synthesis-multimodal-nonblocking} where states are potentially removed from $Q$ to obtain $Q_{S}$: line~8 and line~11.
    Suppose, on line~11, $q' \centernot\in \brs( Q^{k}, \, \{q_{0}\}, \, \trans{}^{-1} )$. Then, by the third item of Lemma~\ref{lem:use-breach}, $\neg \re{q'}$.
    Suppose, on line~8, $q' \in B^{k}$ for some $k$. Then, $q' \in \brs(Q^{k}, \, Q_{\mathrm{blocking}}, \, \dtrans{})$. By Lemma~\ref{lem:use-breach}, this means that $\bad{q'}$. Thus, for all states $q' \centernot\in Q_{S}$, it holds that $\bad{q'} \lor \re{q'}$. If $q' \in Q_{S}'$, this contradicts the assumption that $S'$ is reachable and proper. 
    Therefore, if $q' \centernot\in Q_{S}$, it cannot be that $q' \in Q_{S}'$, because $S'$ would not be proper. Thus, $Q_{S}' \subseteq Q_{S}$.
    \item $\trans{}_{S}' \centernot\subseteq \trans{}_{S}$: Assume $\exists t' \, (t' \in \trans{}_{S}' \land t' \centernot\in \trans{}_{S})$, where $t' = (q_{S,1}', \sigma, q_{S,2}')$ and $q_{S,1}', \, q_{S,2}' \in Q_{S}'$. Transition $t'$ could only have been removed from $\trans{}_{S}$ when $q_{S,1}' \centernot\in Q_{S}$ or $q_{S,2}' \centernot\in Q_{S}$. For every $q \centernot\in Q_{S}$, it has already been shown that either $\bad{q}$ or $\neg\re{q}$. Since $S'$ is proper, then also $q_{S,1}', \, q_{S,2}' \centernot\in Q_{S}'$.
    Thus, it cannot be that $t' \in \trans{}_{S}'$, and therefore $\trans{}_{S}' \subseteq \trans{}_{S}$.
    \item $Q_{m,S}' \centernot= Q_{m,S} \cap Q_{S}'$: Assume $\exists q_{m}' \, (q_{m}' \in Q_{m,S}' \land q_{m}' \centernot\in Q_{m,S} \cap Q_{S}')$. Given that $Q_{S}' \subseteq Q_{S}$, $q_{m}' \in Q_{m,S}'$, and $Q_{m,S}' \subseteq Q_{S}'$, it holds that $q_{m}' \in Q_{S}'$. Then $Q_{m,S}' \centernot= Q_{m,S} \cap Q_{S}'$ only if $q_{m}' \centernot\in Q_{m,S}$. However, in Algorithm~\ref{alg:synthesis-multimodal-nonblocking}, the marking of states is not adjusted. Thus, $Q_{m,S}' = Q_{m,S} \cap Q_{S}'$.
\end{itemize}
It has been shown that $Q_{S}' \subseteq Q_{S}$, $\trans{}_{S}' \subseteq \trans{}_{S}$, and $Q_{m,S}' \subseteq Q_{m,S}$ for a proper supervisor FA $S'$. Therefore, for each proper supervisor FA $S'$ it holds that $S' \sqsubseteq S$.
\end{proof}


\section{Related nonblocking variants} \label{sec:nonblocking-variants}
This section briefly recalls other nonblocking variants and compares them to modal and multimodal nonblocking. It is shown that modal and multimodal nonblocking are unique nonblocking variants. In addition, conditions are provided under which one nonblocking variant implies the other.

\subsection{Standard nonblocking}

Neither standard nonblocking (see Definition~\ref{def:nonblocking-FA}) nor modal nonblocking implies the other.
\begin{proposition}~   
\begin{enumerate}
    \item There exists an FA $P$ and an $\mode \subseteq Q$ such that $\mnb{\mode}{P} \centernot\Rightarrow \snb{P}$.
    \item There exists an FA $P$ and an $\mode \subseteq Q$ such that $\snb{P} \centernot\Rightarrow \mnb{\mode}{P}$.
\end{enumerate}
\end{proposition}

\begin{proof}
1) is shown in Fig.~\ref{fig:example-1}, and 2) is shown in Fig.~\ref{fig:example-2}.
\end{proof}

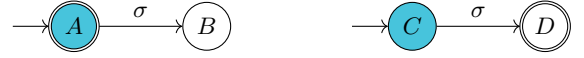
\begin{figure}[H]
\centering
\begin{subfigure}[t]{0.48\linewidth}
    \centering
    \begin{tikzpicture}[node distance=5.0em, font=\small]
        \node [state1, initial, accepting] (A) {$A$};
        \node [state, right of=A] (B) {$B$};
        \draw (A) edge [c] node [above] {$\sigma$} (B);
    \end{tikzpicture}
    \caption{$\mode = \{A\}$ and $Q_{m} = \{A\}$.}
    \label{fig:example-1}
\end{subfigure}\hfill
\begin{subfigure}[t]{0.48\linewidth}
    \centering
    \begin{tikzpicture}[node distance=5.0em, font=\small]
        \node [state1, initial] (C) {$C$};
        \node [state, right of=C, accepting] (D) {$D$};
        \draw (C) edge [c] node [above] {$\sigma$} (D) ;
    \end{tikzpicture}
    \caption{$\mode = \{C\}$ and $Q_{m} = \{D\}$.}
    \label{fig:example-2}
\end{subfigure}
\caption{Example FAs.}
\label{fig:examples}
\end{figure}

When $\mmode = \{ \mode \}$, the same examples can be used to show that neither standard nonblocking nor multimodal nonblocking implies the other.

While there is no general implication between standard nonblocking and modal nonblocking, there does exist an implication relation for the special case where $\mode = Q$.
\begin{corollary}
\label{cor:standard-nonblocking-comparison}
Let $P = $ \automaton be an FA. Then, $\mnb{Q}{P} \Leftrightarrow \snb{P}$.
\end{corollary}

\begin{proof}
Definitions~\ref{def:nonblocking-FA} and~\ref{def:modal-nonblocking-FA} coincide when $\mode=Q$.
\end{proof}

The same holds for multimodal nonblocking where $\mmode = \{Q\}$. In addition, the case where $\mmode$ covers $Q$ is considered.

\begin{proposition}
Let $P = $ \automaton be an FA, and let $\mmode \subseteq 2^{Q}$, where $\bigcup_{\mode \in \mmode} (\mode) = Q$. Then, $\mmnb{P} \Rightarrow \snb{P}$.
\end{proposition}

\begin{proof}
Because $\bigcup_{\mode \in \mmode} (\mode) = Q$, for every $q \in Q$ also $q \in \mode$ for some \mode. Given that $P$ is modal nonblocking for every $\mode \in \mmode$, for every reachable state in \mode there must exist a path to a marked state that stays within \mode.
By Lemma~\ref{lem:transs-transsM}, there generally exists a path to a marked state, and therefore $P$ must be standard nonblocking.\\
The opposite implication does not hold. This is shown by the FA in Fig.~\ref{fig:example-2} when $\mmode = \{ \{ C \}, \, \{ D \} \}$. For this $P$, it holds that \snb{P}, but not \mmnb{P}.
\end{proof}


\subsection{Generalized nonblocking}

For generalized nonblocking, a subset of states $Q' \subseteq Q$ is defined from which a marked state should be reachable.
\begin{definition}[Generalized nonblocking FA~\cite{Malik2008GeneralisedNonblocking}]
\label{def:generalized-nonblocking-automaton}
Let $P = $ \automaton be an FA, and let $Q' \subseteq Q$ be a state set. $P$ is \emph{generalized nonblocking} for $Q'$, denoted \gnb{Q'}{P}, iff $\forall q \in Q' \, (\re{q} \Rightarrow \snb{q})$.
\end{definition}

Neither generalized nonblocking nor modal nonblocking implies the other.
\begin{proposition}~   
\begin{enumerate}
    \item There exists an FA $P$, an $\mode \subseteq Q$ and a $Q' \subseteq Q$ such that $\mnb{\mode}{P} \centernot\Rightarrow \gnb{Q'}{P}$.
    \item There exists an FA $P$, an $\mode \subseteq Q$ and a $Q' \subseteq Q$ such that $\gnb{Q'}{P} \centernot\Rightarrow \mnb{\mode}{P}$.
\end{enumerate}
\end{proposition}
\begin{proof}
Standard nonblocking is a special case of generalized nonblocking where $Q'= Q$. Therefore, the examples from Fig.~\ref{fig:examples} show 1) and 2) for $Q' = Q$.
\end{proof}

While there is no general implication between generalized nonblocking and modal nonblocking, there is an implication relation for the special case where $Q' = \mode$.
\begin{proposition}
Let $P = $ \automaton be an FA, and let $Q' \subset Q$. Then, $\mnb{Q'}{P} \Rightarrow \gnb{Q'}{P}$.
\end{proposition}
\begin{proof}
It must be proven that \gnb{Q'}{P}, or by Definition~\ref{def:generalized-nonblocking-automaton}: $\forall q \in Q' \, ( \re{q} \Rightarrow \snb{q} )$. Assume \mnb{Q'}{P}, so by Definitions~\ref{def:modal-nonblocking-state} and~\ref{def:modal-nonblocking-FA}: $\forall q \in Q \, (\re{q} \Rightarrow ( q \in Q' \Rightarrow \exists q_{m} \in Q_{m} \, ( q \transs{}_{Q'} q_{m}) ) )$. Thus, for all reachable states $q \in Q'$ it holds that $\exists q_{m} \in Q_{m} \, ( q \transs{}_{Q'} q_{m})$. By Lemma~\ref{lem:transs-transsM}, then also $\exists q_{m} \in Q_{m} \, (q \transs{} q_{m})$, denoted \snb{q}. This means \snb{q} for all reachable states in $Q'$, which was to be proven.\\
The opposite implication does not hold.
This is shown by the FA in Fig.~\ref{fig:example-faults} when $Q' = \{ A, \, B, \, E\}$. For this $P$, it holds that \gnb{Q'}{P}, but not \mnb{Q'}{P}.
\end{proof}


\subsection{Progressive nonblocking}

For progressive nonblocking, a set of progressive events $\Sigma_{p} \subseteq \Sigma$ is defined. Only progressive events may be used to reach a marked state. 
\begin{definition}[Progressive nonblocking FA~\cite{Ware2014ProgressiveVerification}]
\label{def:progressive-nonblocking-automaton}
Let $P = $ \automaton be an FA, let $\Sigma_{p} \subseteq \Sigma$ be a set of progressive events, and let $\trans{}_{p} = \trans{} \cap Q \times \Sigma_{p} \times Q$ be the set of transitions labeled with progressive events. $P$ is \emph{progressive nonblocking} for $\Sigma_{p}$, denoted \pnb{P}, iff $\forall q \in Q \, (\re{q} \Rightarrow \exists q_{m} \in Q_{m} \, (q \transs{}_{p} q_{m}) )$.
\end{definition}

Neither progressive nonblocking nor modal nonblocking implies the other.
\begin{proposition}~   
\begin{enumerate}
    \item There exists an FA $P$, an $\mode \subseteq Q$ and a $\Sigma_{p} \subseteq \Sigma$ such that $\mnb{\mode}{P} \centernot\Rightarrow \pnb{P}$.
    \item There exists an FA $P$, an $\mode \subseteq Q$ and a $\Sigma_{p} \subseteq \Sigma$ such that $\pnb{P} \centernot\Rightarrow \mnb{\mode}{P}$.
\end{enumerate}
\end{proposition}
\begin{proof}
Standard nonblocking is a special case of progressive nonblocking where $\Sigma_{p} = \Sigma$. Therefore, the examples from Fig.~\ref{fig:examples} show 1) and 2) for $\Sigma_{p} = \Sigma$.
\end{proof}

While there is no general implication between progressive nonblocking and modal nonblocking, there is an implication under a specific condition.
\begin{proposition}
\label{prp:progressive-modal}
For an FA $P$, mode $\mode \subseteq Q$ and $\Sigma_{p} \subseteq \Sigma$ such that $\forall ( q, \, \sigma, \, q' ) \in \trans{} \, (q \in \mode \Rightarrow ( q' \in \mode \Leftrightarrow \sigma \in \Sigma_{p} ) )$, it holds that $\mnb{\mode}{P} \Rightarrow \pnb{P}$.
\end{proposition}
\begin{proof}
The construction of $\Sigma_{p}$ is such that for any transition $( q, \, \sigma, \, q' ) \in \trans{}$ where $q \in \mode$, it always holds that $q' \in \mode \Leftrightarrow \sigma \in \Sigma_{p}$: 
\begin{itemize}
    \item when $q'$ also in \mode, then $\sigma \in \Sigma_{p}$;
    \item when $q'$ not in \mode, then $\sigma \centernot\in \Sigma_{p}$.
\end{itemize}
Thus, all transitions leaving \mode are not in $\Sigma_{p}$, and all other transitions are in $\Sigma_{p}$. Then, $\mnb{\mode}{P} \Rightarrow \pnb{P}$.\\
The opposite implication does not hold, which is shown by the FA in Fig.~\ref{fig:example-1} when $\Sigma_{p} = \emptyset$. The FA is modal nonblocking and the transitions leaving \mode are not in $\Sigma_{p}$, yet from state $B$ no marked state can be reached, and therefore the FA is not progressive nonblocking.
\end{proof}

Using Proposition~\ref{prp:progressive-modal} on Example~1,
it can be shown that progressive nonblocking with $\Sigma_{p} = \{\sigma_{1}, \, \sigma_{2}, \, \sigma_{\mathrm{fix}} \}$ is equal to modal nonblocking with $\mode = \{ A, \, B, \, E \}$. Indeed, by leaving the event $\sigma_{\mathrm{error}}$ out of the set of progressive events, progressive nonblocking effectively requires that \mode may not be left to reach a marked state.

However, progressive nonblocking cannot be applied to Examples~2 and~3 because they feature multiple state sets of interest. This would require multiple sets of progressive events to be defined, which is not considered in~\cite{Ware2014ProgressiveVerification}.


\subsection{Multitasking nonblocking}

For multitasking nonblocking, sets of colored marked states are defined. A state is considered multitasking nonblocking iff, for each color, there exists a path towards a marked state of that color.
\begin{definition}[Multitasking nonblocking FA~\cite{DeQueiroz2005MultitaskingSystems}]\label{def:multitasking-nonblocking-FA}
Let $P = $ \automaton be an FA, and let $\mathcal{C} \subseteq 2^{Q}$ contain sets of colored marked states $C$. $P$ is \emph{multitasking nonblocking} for $\mathcal{C}$, denoted $\mtnb{P}$, iff $\forall q \in Q \, (\re{q} \Rightarrow \forall C \in \mathcal{C} \, ( \exists q' \in C \, (q \transs{} q') ) )$.
\end{definition}

Neither multitasking nonblocking nor modal nonblocking implies the other.
\begin{proposition}~   
\begin{enumerate}
    \item There exists an FA $P$, an $\mode \subseteq Q$ and a $\mathcal{C} \subseteq 2^{Q}$ such that $\mnb{\mode}{P} \centernot\Rightarrow \mtnb{P}$.
    \item There exists an FA $P$, an $\mode \subseteq Q$ and a $\mathcal{C} \subseteq 2^{Q}$ such that $\mtnb{P} \centernot\Rightarrow \mnb{\mode}{P}$.
\end{enumerate}
\end{proposition}
\begin{proof}
Standard nonblocking is a special case of multitasking nonblocking, where $\mathcal{C}= \{Q_{m}\}$. Therefore, the examples from Fig.~\ref{fig:examples} show 1) and 2) for $\mathcal{C}= \{Q_{m}\}$.
\end{proof}

There are, in fact, no non-trivial conditions where there is an implication between multitasking nonblocking and modal, or multimodal, nonblocking. In other words, the variants are complimentary.
Multitasking nonblocking guarantees that multiple colored marked states in different modes remain reachable, while multimodal nonblocking guarantees for each mode that a marked state can be reached while staying within that mode.

Their complimentary nature can be shown in Examples~2 and~3. In Example~2, one may choose $\mathcal{C} = \{ \{A_{2}\}, \, \{B_{1}\} \}$ to guarantee that a colored marked state in both modes can always be reached. In addition, multimodal nonblocking guarantees that each mode itself is nonblocking without leaving that mode.
In Example~3, one may choose $\mathcal{C} = \{ \{GR\}, \, \{RG\} \}$, to guarantee that the green state of either traffic light is always reachable. In addition, multimodal nonblocking guarantees that both traffic lights are able to turn green when considered separately, i.e., one light can turn green without depending on the other turning green.


\subsection{Quantitative nonblocking}

For quantitative nonblocking, a positive integer $k$ is defined. A state is quantitative nonblocking if there exists a sequence of transitions towards a marked state consisting of $k$ or less transitions.
\begin{definition}[Quantitative nonblocking FA~\cite{Zhang2024QuantitativelySystem}]\label{def:quantitative-nonblocking-FA}
Let $P = $ \automaton, and let $q \transs{}^{k} q'$ denote that $q \transs{} q'$ in at most $k \in \mathbb{N}_{0}$ transitions. $P$ is \emph{quantitative nonblocking} for $k$, denoted \qnb{P}, iff $\forall q \in Q \, (\re{q} \Rightarrow \exists q_{m} \in Q_{m} \, ( q \transs{}^{k} q_{m} ) )$.
\end{definition}

Neither quantitative nonblocking nor multitasking nonblocking implies the other.
\begin{proposition}~   
\begin{enumerate}
    \item There exists an FA $P$, an $\mode \subseteq Q$ and a $k \in \mathbb{N}_{0}$ such that $\mnb{\mode}{P} \centernot\Rightarrow \qnb{P}$.
    \item There exists an FA $P$, an $\mode \subseteq Q$ and a $k \in \mathbb{N}_{0}$ such that $\qnb{P} \centernot\Rightarrow \mnb{\mode}{P}$.
\end{enumerate}
\end{proposition}
\begin{proof}
Standard nonblocking is a special case of quantitative nonblocking where $k = |Q|$.
Therefore, the examples from Fig.~\ref{fig:examples} show 1) and 2) for $k = |Q|$.
\end{proof}

While there is no general implication between quantitative nonblocking and modal nonblocking, there is an implication for multimodal nonblocking under a specific condition.

\begin{proposition}
For an FA $P$, integer $k$ and sets of desired states $\mmode \subseteq 2^{Q}$, constructed such that each \mode contains states reachable in $k$ or less transitions from a state $q$: $\mode(q) = \{ q' \in Q | q \transs{}^{k} q' \}$ and $\mmode = \{ \mode(q) | q \in Q \}$, it holds that $\mmnb{P} \Rightarrow \qnb{P}$.
\end{proposition}

\begin{proof}
Replace both sides with Definitions~\ref{def:multimodal-nonblocking-FA} and~\ref{def:quantitative-nonblocking-FA}: $\forall q \in Q \, (\re{q} \Rightarrow \mmnb{q} ) \Rightarrow \forall q \in Q \, (\re{q} \Rightarrow \exists q_{m} \in Q_{m} \, ( q \transs{}^{k} q_{m} ) )$. Let $q \in Q$ and assume \re{q} to reduce the proof to $\mmnb{q} \Rightarrow \exists q_{m} \in Q_{m} \, ( q \transs{}^{k} q_{m} ) $.

Assume \mmnb{q}, or by Definitions~\ref{def:modal-nonblocking-state} and~\ref{def:multimodal-nonblocking-state}: $\forall \mode \in \mmode (q \in \mode \Rightarrow \exists q_{m} \in Q_{m} \, (q \transs{}_{\mode} q_{m}) )$.
Thus, for every $\mode \in \mmode$ and $q \in \mode$, it holds that $\exists q_{m} \in Q_{m} \, (q \transs{}_{\mode} q_{m})$. By the construction of \mode, then also $\exists q_{m} \in Q_{m} \, ( q \transs{}^{k} q_{m} )$.\\
The opposite implication does not hold, because \qnb{P} does not imply that all states $q \in \mode$ are modal nonblocking for every $\mode \in \mmode$.
\end{proof}


\section{Conclusion} \label{sec:conclusion}

This paper has introduced modal and multimodal nonblocking to capture a more expressive and versatile nonblocking property. The novel variants can be applied to systems with fault behavior, reconfiguration and multiple control goals. Synthesis algorithms are presented to obtain modal or multimodal nonblocking supervisors. The novel variants are compared to established nonblocking variants.

With modal and multimodal nonblocking, supervisory control synthesis could be sensibly applied to safety-critical water infrastructures with multiple control goals, like in \cite{Minkenberg2026AbstractionInsights}. This does require some implementation of the novel nonblocking variants in a tool like ESCET~\cite{Fokkink2023EclipseToolkit}. 

Additionally, for practical applications of the novel nonblocking variants, it is necessary that the supervisor is also \emph{safe}: the supervisor never violates a given requirements model.
The safety property is left out in this paper, but could be introduced by adding a \emph{plantified} requirement model, see~\cite{Flordal2007CompositionalEquivalence}. Synthesis of the synchronous composition~\cite{Cassandras2008} of the plant model and the plantified requirement model produces a safe supervisor by default. The hypothesis is that this still holds in the context of modal nonblocking, as long as the sink state belongs to at least one \mode.

Furthermore, since EFAs can be transformed to FAs, as shown in~\cite{Skoldstam2008SupervisoryVariables-Revised}, the modal and multimodal nonblocking property can be defined for EFAs as well. The expansion to EFAs is future work. 

Finally, there is potential for combining different, complementary nonblocking variants in a real case study. Specifically, progressive nonblocking, multitasking nonblocking and multimodal nonblocking could be combined to achieve a very versatile way to obtain very specific nonblocking properties.


\bibliographystyle{IEEEtran.bst}
\bibliography{refs.bib}


\end{document}